\documentclass[11pt,reqno]{amsart}
\usepackage{graphicx,hyperref}

\newtheorem{theorem}{THEOREM}
\newtheorem{Lemma}{Lemma}
\newtheorem{Proposition}{PROPOSITION}
\newtheorem{corollary}{COROLLARY}
\theoremstyle{definition}
\newtheorem{definition}{DEFINITION}
\theoremstyle{definition}
\newtheorem{Remark}{Remark}

\newcommand{\R}{\mathbb{R}}
\renewcommand{\S}{\mathbb{S}}

\newcommand{\C}{\mathbb{C}}

\newcommand{\F}{\mathcal{F}}
\newcommand{\W}{\mathcal{W}}

\newcommand{\gm}{\gamma}
\newcommand{\al}{\alpha}

\newcommand{\bra}{\langle}
\newcommand{\ket}{\rangle}

\newcommand{\be}{\begin{equation}}
\newcommand{\ee}{\end{equation}}
\newcommand{\bea}{\begin{align}}
\newcommand{\eea}{\end{align}}

\newcommand\de{\mathcal D}
\newcommand\infspec{{\rm{inf\, spec\,}}}
\newcommand\eps\epsilon
\newcommand\V{\mathcal{V}}
\newcommand\B{\mathcal{B}}

\newcommand\om{\overline{m}}

\DeclareMathOperator{\Tr}{Tr}

\numberwithin{equation}{section}

\begin{document}

\title{The BCS gap equation for spin-polarized fermions}

\author[A. Freiji]{Abraham Freiji} \address{{\rm (Abraham Freiji)},
  Department of Neurology, UAB, SC 350, 1530 3rd Ave S, Birmingham, AL
  35294, USA} \email{afreiji@uab.edu}

\author[C. Hainzl]{Christian Hainzl} \address{{\rm (Christian Hainzl)}
  Mathematisches Institut, Universit\"at T\"ubingen, Auf der
  Morgenstelle 10, 72076 T\"ubingen, Germany}
\email{christian.hainzl@uni-tuebingen.de}

 \author[R. Seiringer]{Robert Seiringer} \address{{\rm (Robert Seiringer)}, Department of
   Mathematics and Statistics, McGill University, 805 Sherbrooke
   Street West, Montreal, QC, H3A 2K6, Canada}
 \email{rseiring@math.mcgill.ca}

\thanks{\copyright\ 2011 by the authors. This paper may be reproduced, in its entirety, for non-commercial purposes. \\\indent Partial support by U.S. National Science
Foundation grants DMS-0800906 (C.H.) and PHY-0845292 (R.S.) and the NSERC (R.S.) is
gratefully acknowledged.}

\begin{abstract}
  We study the BCS gap equation for a Fermi gas with unequal
  population of spin-up and spin-down states. For $\cosh(\delta_\mu/T)
  \leq 2$, with $T$ the temperature and $\delta_\mu$ the chemical
  potential difference, the question of existence of non-trivial
  solutions can be reduced to spectral properties of a linear
  operator, similar to the unpolarized case studied previously in
  \cite{FHNS,HHSS,HS}.  For $\cosh(\delta_\mu/T) > 2$ the phase
  diagram is more complicated, however. We derive upper and lower
  bounds for the critical temperature, and study their behavior in the
  small coupling limit.
\end{abstract}

\maketitle

\section{Introduction}

Spin-polarized fermionic systems have attracted a lot of interest in
recent years, and have been intensely studied both experimentally and
theoretically. They are of relevance in various areas of physics,
ranging from cold gases to nuclear stars.

The goal of our present paper is to study the BCS gap equation for
such imbalanced systems. Experiments \cite{ZSSW,ZSSK} have
demonstrated a significantly richer phase diagram than in the
balanced, unpolarized case. The possibility of pairing in polarized 
systems was first pointed out in \cite{sarma}, and since then has been
intensively studied in the literature, mostly under the assumption of
a  contact interaction  among the fermions. We
refer to \cite{Chen} and \cite{zwerger} for recent reviews on this
subject.

Our work is a continuation of recent papers \cite{FHNS,HHSS,HS,HS3}
where the BCS gap equation in the balanced case for systems with
general pair interaction potential $V$ was investigated. In
particular, a criterion for potentials giving rise to superfluidity
was derived. The form of the interaction in actual physical systems
can be quite general, and hence it is important to keep $V$ as general
as possible.  In the following we recall these results before we
present the main results in the imbalanced case.

\subsection{Balanced fermionic systems}\label{sec:bal}

Consider a gas consisting of neutral spin $\frac 12$ fermions. The
kinetic energy of these particles is described by the non-relativistic
Schr\"odinger operator, and their interaction by a pair potential which we write for convenience as 
$2\lambda V$, with $\lambda$ being a coupling parameter.  According to
Bardeen, Cooper and Schrieffer (BCS) \cite{BCS}  a superfluid state of the system 
is characterized by the existence of a non-trivial solution of the
{\em gap equation}
\begin{equation}\label{bcseintro}
\Delta(p) = -\frac \lambda{(2\pi)^{3/2}} \int_{\R^3} \hat V(p-q)
\frac{\Delta(q)}{E(q)} \tanh \frac{E(q)}{2T} \, dq
\end{equation}
at some temperature $T\geq 0$, with $E(p)= \sqrt{(p^2-\mu)^2 +
  |\Delta(p)|^2}$. Here, $\mu \in \R$ is the chemical potential and $\hat
V(p)=(2\pi)^{-3/2}\int_{\R^3} V(x) e^{-ipx} dx $ denotes the Fourier
transform of $V$. The function $\Delta(p)$ is the order parameter which is closely related
to the wavefunction of the {\em Cooper pairs}.

The detailed investigation of Eq.~\eqref{bcseintro}
 in \cite{FHNS,HHSS,HS,HS3,HS2,HS4,HS5} is based on the
observation in \cite{HHSS} that the existence of a non trivial
solution of \eqref{bcseintro} is equivalent to the existence of a
negative eigenvalue of the {\em linear } operator $ K_T(-i\nabla) + \lambda V(x)$
on $L^2(\R^3)$, where 
$$ 
K_T(p) = \frac {|p^2 - \mu|}{\tanh \frac{|p^2 -
      \mu|}{2T}}\,. $$ 
The critical temperature $T_c(\lambda V)$  is thus defined by the equation 
\begin{equation}\label{def:tc}
\infspec \left( K_{T_c} + \lambda V \right) =0\,.
\end{equation}
Since $K_T(p)$ is pointwise monotone in $T$, this defines $T_c$
uniquely. For $T < T_c(\lambda V)$ the equation \eqref{bcseintro} has
a non-trivial solution, while for $T \geq T_c(\lambda V)$ it has
none.

For $\mu>0$, $K_T(p)$ attains its minimum on the sphere $\Omega_\mu = \{
p\, : \, p^2 =\mu\}$. In the small coupling limit $\lambda\to 0$, the
study of $T_c$ reduces to the study of an effective operator $\V_\mu :
L^2(\Omega_\mu) \to L^2(\Omega_\mu)$, given by
\begin{equation}\label{VV}
  \big(\V_{\mu} u\big)(p) =
  \frac 1{(2\pi)^{3/2}} \frac 1{\sqrt{\mu}}\int_{\Omega_{\mu}}\hat V(p-q) u(q)
  \,d\omega(q)\,,
\end{equation}
with $d\omega$ denoting the Lebesgue measure on the sphere.  The
lowest eigenvalue $e_\mu<0$ of $\V_{\mu}$ determines the leading order of
$T_c(\lambda V)$ as $\lambda \to 0$ \cite{FHNS}. In fact, 
in \cite{FHNS,HS} it is shown that 
\begin{equation}\label{form:tc}
  T_c = \mu  \left(\frac{8 e^{\gamma-2}}{\pi} + o(1)\right)  e^{\pi/(2 \sqrt{\mu} b_\mu)} \quad \text{as $\lambda \to 0$}\,,
\end{equation}
where $\gamma\approx 0.577$ denotes Euler's constant, and $b_\mu = \pi \lambda
e_\mu/(2\sqrt\mu ) + O(\lambda^2)$, where the $O(\lambda^2)$ term is determined
by another effective operator of $L^2(\Omega_\mu)$.  The parameter
$b_\mu$ can be interpreted as an effective scattering length (in
second order Born approximation), which reduces to the usual
scattering length of the interaction potential $2\lambda V$ in the low density
limit, i.e., as  $\mu\to 0$ \cite{HS3}.

\section{Main results}

We consider a gas of spin $\frac12$ particles at temperature $T\geq 0$, interacting via a
pair interaction potential $2\lambda V$. We assume the two spin states
are unequally populated, but have equal masses. We introduce two
separate chemical potentials, $\mu_+$ and $\mu_-$, for the spin up and
spin down particles, respectively. We denote the average chemical
potential by $\bar\mu=(\mu_+ + \mu_-)/2$, and half their
difference by
\begin{equation}
\delta_\mu = \frac{\mu_+ - \mu_-}2\,,
\end{equation}
which we assume to be non-negative without loss of generality. 

Let $\gamma_+(p)$ and $\gamma_-(p)$ denote the momentum distributions
for the spin-up and spin-down fermions respectively.  In addition, let
$\alpha$ denote the Cooper pair wave-function.  The BCS functional
$\F_T$ is defined as
\begin{align}\nonumber
\F_T(\gm_{+},\gm_{-},\al) & = \frac{1}{2} \int_{\R^3}
(p^2-\mu_{+})\gm_{+}(p)dp+ \frac{1}{2} \int_{\R^3}
(p^2-\mu_{-})\gm_{-}(p)dp   \\  & \quad  +\lambda \int_{\R^3}
|\alpha(x)|^2 V(x)dx -\frac T2 S(\gamma_+,\gamma_-,\alpha)\,,\label{freeenergy0}
\end{align}
where the entropy $S$ is 
\begin{equation}\label{def:entropy}
S(\gamma_+,\gamma_-,\alpha) = - \int_{\R^3} \Tr_{\C^4}\left[\Gamma(p)
  \ln \Gamma(p)\right]dp, \,\, 
\end{equation}
with 
\begin{equation}\label{def:Gamma}
\Gamma(p)= \left(\begin{matrix}\gamma_+(p) & 0 & 0 & \hat \alpha(p)\\
0 & \gm_-(p) & -\hat \alpha(p) & 0 \\
0 & -\overline{\hat \alpha(p)} & 1-\gm_+(p) & 0\\
\overline{\hat \alpha(p)} & 0 & 0 & 1-\gm_-(p)
\end{matrix}\right)\,.
\end{equation}
The functions $\gamma_\pm$ and $\alpha$ are constrained by demanding that $0 \leq \Gamma(p) \leq 1$ as an operator on $\C^4$, for all $p\in\R^3$.

The functional $\F_T$ can be heuristically derived \cite{Leg,HHSS} by
starting from a genuine many-body Hamiltonian and making three steps
of simplifications. First, one considers only quasi-free (or
generalized Hartree-Fock \cite{BLS}) states. Second, one ignores the
exchange and direct terms in the interaction energy. Finally, one
assumes translation invariance and works out the energy per unit
volume.  The specific ansatz of the states in the form of $\Gamma$
follows from assuming, in addition, that the total spin points in the
direction of the applied field (quantified by $\delta_\mu$).

We will show below that on an appropriate domain 
the functional $\F_T$ attains its minimum and the corresponding
Euler-Lagrange equation, expressed in terms of the {\em order
parameter} $\Delta  = - 2 \lambda \hat V \ast \hat \alpha$, 
takes the form
\begin{equation}\label{overgapeqn}
 \Delta = -\lambda \hat V \ast \frac{\Delta}{2E}  \left(\tanh
\left(\frac{E +\delta_\mu}{2T} \right)+\tanh
\left(\frac{E-\delta_\mu}{2T} \right)\right)\,,
 \end{equation}
 where $E(p) = \sqrt{(p^2 - \bar\mu)^2 + |\Delta(p)|^2}$. For
 convenience we define the convolution with a factor $(2\pi)^{-3/2}$
 in front, i.e., $f\ast g(p) = (2\pi)^{-3/2} \int_{\R^3} f(p-q)g(q)
 dq.$

In terms of
\begin{equation}\label{def:kD}
 K^{\Delta(p)}_{T, \delta_\mu}(p)=\frac{2E(p)}{\tanh \left(\frac{E(p)+\delta_\mu}{2T} \right)+
\tanh \left(\frac{E(p)-\delta_\mu}{2T} \right)}\,,
\end{equation}
 the gap equation can
conveniently be expressed as $\Delta = - \lambda \hat V \ast
(\Delta/K^{\Delta}_{T, \delta_\mu})$. Define also
$$
\widetilde K_{T, \delta_\mu}(p) = \inf_{x>0} K^{x}_{T, \delta_\mu}(p)\,.
$$
The basis of our analysis is the following theorem.

\begin{theorem}\label{over1}
Let $V\in L^{3/2}(\R^3) $ be real-valued, $\bar\mu \in\R$, $T\geq 0$ and $\delta_\mu \geq 0$.
\begin{itemize}
\item[(i)] If the linear operator $K^{0}_{T, \delta_\mu}(-i\nabla)
  +\lambda V(x)$ has a negative eigenvalue, then the functional $\F_T$
  attains its minimum at some $\alpha \neq 0$.
\item[(ii)] If $\F_T$ attains its minimum at some $\alpha \neq 0$, then
$\Delta(p) = - 2\lambda \hat V \ast \hat \alpha (p)$ satisfied the
BCS gap equation \eqref{overgapeqn}.
\item[(iii)]
If the BCS equation has a non-trivial solution, then
the linear operator $\widetilde K_{T, \delta_\mu}(-i\nabla) +\lambda V(x)$ has a
negative eigenvalue. 
\end{itemize}
\end{theorem}

The proof of Theorem~\ref{over1} is given in Section~\ref{over1sec}.

\begin{Remark}\label{rem1}
  We shall show in Corollary \ref{cor1} that for all pairs
  $(\delta_\mu,T)$ with $\cosh(\delta_\mu / T) \leq 2$, the function
  $K^{\Delta}_{T, \delta_\mu}$ is pointwise monotone in $\Delta$ and hence
$$\widetilde K _{T, \delta_\mu}(p) = K^{0}  _{T, \delta_\mu}(p) \quad \forall p.$$  This
means that statements (i)$-$(iii) in Theorem~\ref{over1} are
actually {\em equivalent}, exactly as in the balanced case
\cite[Thm.~1]{HHSS}. Hence we have obtained a linear
characterization for the existence of a non-trivial solution of the
non-linear BCS gap equation. This does not hold if $ \cosh(\delta_\mu
/ T) > 2$, however, as discussed below.
\end{Remark}

\subsection{Critical temperatures}\label{ss:tc}

Since we do not have equivalence of the statements (i)$-$(iii) in Theorem~\ref{over1} in
case $\cosh(\delta_\mu / T) >2$, we do not have a simple definition of the
critical temperature as in (\ref{def:tc}) in this case.  With the aid
of the operators $K^0_{T,\delta_\mu}$ and $\widetilde
K_{T,\delta\mu}$, we can define curves $T^{\rm i}_{\delta_\mu}$ and
$T^{\rm o}_{\delta_\mu}$ in the $(\delta_\mu,T)$ plane, given by
\begin{align}\label{def:ti}
T^{\rm i}_{\delta_\mu}(\lambda V) : \,\, & \{(\delta_\mu,T)\, | \, {\rm
inf \, spec}
(K^0_{T,\delta_\mu} + \lambda V ) = 0\},\\ \label{def:to}
T^{\rm o}_{\delta_\mu}(\lambda V): \,\, & \{(\delta_\mu,T)\, | \, {\rm inf
\, spec} (\widetilde K_{T,\delta_\mu} + \lambda V ) = 0\}.
\end{align}
Below the curve $T^{\rm i}_{\delta_\mu}$ one has $ \infspec
(K^0_{T,\delta_\mu} + \lambda V ) < 0$, and hence Theorem~\ref{over1}(i) implies that the functional $\F_T$ has a global minimum at some
$\alpha \neq 0$.  Above the curve $T^{\rm o}_{\delta_\mu}$, however, Theorem
\ref{over1}(iii) implies that the gap equation does not have a
non-trivial solution. Consequently the minimum of $\F_T$ is attained
for $\alpha = 0$, the {\em normal} state. See Fig.~\ref{fig1}.

In the region between $T^{\rm i}_{\delta_\mu}$ and $T^{\rm o}_{\delta_\mu}$, there
can exist non-trivial solutions to the gap equation even if $\alpha$
vanishes identically for the minimizer of $\F_T$.  In fact it was
shown in \cite{CalMo} by numerical calculations, using a contact
interaction potential, that there is a parameter regime where the gap
equation has a solution but the corresponding energy is higher than
that of the normal state.

Using the methods developed in \cite{FHNS, HS} we are able to evaluate
$T^{\rm i}_{\delta_\mu}(\lambda V)$ and $T^{\rm
  o}_{\delta_\mu}(\lambda V)$ in the small coupling limit $\lambda \to
0$, and obtain analogous formulas as in the balanced case in
(\ref{form:tc}). With $T_c$ denoting the critical temperature in the balanced
case (at $\mu=\bar\mu$), we shall show that
\begin{equation}\label{claim:tc}
T_{\delta_\mu}^\#  = T_c  e^{-\kappa^\#(\delta_{\mu}/T_{\delta_\mu}^\#)}
\end{equation}
as $\lambda\to 0$, with $\#$ standing for either i or o, and the
$\kappa^\#$ are two explicit non-negative functions (defined in
\eqref{kappai} and \eqref{kappao}, respectively), with
$\kappa^\#(0)=0$, of course.  We refer to Theorems~\ref{lowerb} and
\ref{upperb} in Section~\ref{sec:temp} for the precise statements. The
resulting curves are plotted in Figure~\ref{fig1}. 

\begin{figure}[h]\centering\includegraphics[width=.9\textwidth]{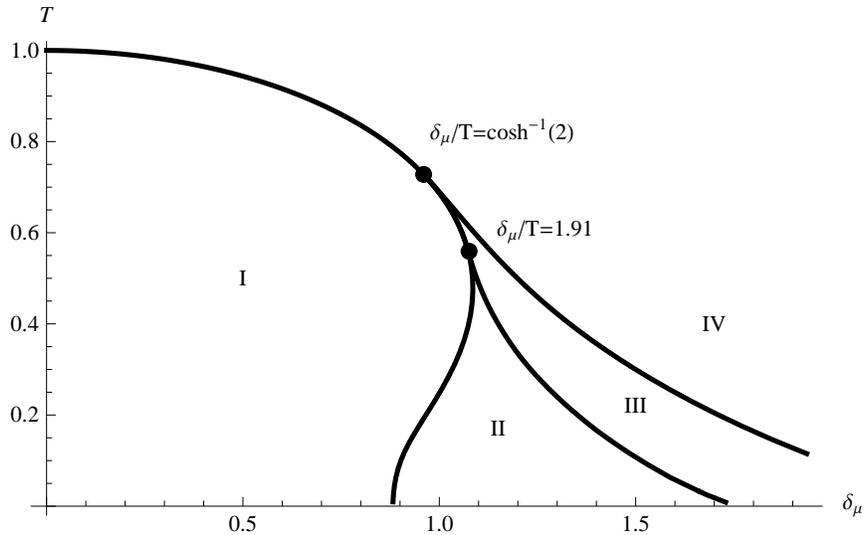}
  \caption{The curves $T^{\rm i}_{\delta_\mu}$, $T^{\rm g}_{\delta_\mu}$ and
    $T^{\rm o}_{\delta_\mu}$ (in increasing order), in units of $T_c$, the
    critical temperature in the balanced case.  In region I the system
    is in a superfluid phase. In region IV, there are no non-trivial
    solutions to the BCS gap equation, and the system is the normal
    phase. In region III we prove, under additional assumptions on
    $V$, that the system is in a normal phase, even though non-trivial
    solutions to the gap equation might
    exist.}\label{fig1}\end{figure}

\subsubsection{A more detailed phase diagram}\label{ss:det}

Under additional assumptions on the interaction potential $V$ we are
able to derive a sharper upper bound on the critical temperature, at
least for small coupling.  Let $T^{\rm g}_{\delta_\mu}$ be the curve
defined by
\begin{equation}\label{tcdef}
T^{\rm g}_{\delta_\mu}(\lambda V)=\{(\delta_\mu,T):\inf_{y>0} \inf {\rm spec}
(K^{y}_{T,\delta_\mu}+ \lambda V)=0\}.
\end{equation}
Note the difference with $T^{\rm o}_{\delta_\mu}$ in (\ref{def:to}), where
the minimum is taken pointwise for every fixed $p$ in $K^y_{T,\delta_\mu}(p)$.
Recall also from Remark \ref{rem1} that $K^{y}_{T,\delta_\mu}$ is not
monotone in $y$ if $\cosh\left(\delta_\mu/T \right) > 2$, so the
minimum of the spectrum of $K^{y}_{T,\delta_\mu}+ \lambda V$ is not
necessarily attained at $y=0$.

In terms of $\alpha$ the BCS gap equation
\eqref{overgapeqn} can conveniently be written as
\begin{equation}
\left( K^\Delta_{T,\delta_\mu} + \lambda V\right) \alpha = 0.
\end{equation}
Let us {\em assume} that $ K^\Delta_{T,\delta_\mu} + \lambda V$ has
$0$ as its lowest eigenvalue. For minimizers of the functional $\F_T$,
this is the case, for instance, of $\hat V$ is negative, since in this
case $\hat \alpha$ is negative and hence the ground state, by a
Perron-Frobenius type argument.

In this case, we will show that $\Delta$ can be treated as a constant
in the small coupling limit, with its value given by the one it
attains on the Fermi surface. As a consequence, we shall show that if
$(\delta_\mu,T)$ lies above $T^{\rm g}_{\delta_\mu}$, the system is in
a normal phase (in the sense that $\alpha=0$ for the minimizer of
$\F_T$), even though the BCS gap equation can have non-trivial
solutions. Hence the sharper upper bound $T^{\rm g}_{\delta_\mu}$ on
the critical temperature for superfluidity holds in place of $T^{\rm
  o}_{\delta_\mu}$; see Fig.~\ref{fig1}.  The precise formulation of
our results is given in Theorem~\ref{tctheorem} in
Section~\ref{sec:temp}.

Numerical calculations in \cite{CalMo} (see, in particular, Fig.~2
there) predict that between the curves $T^{\rm i}_{\delta_\mu}$ and
$T_{\delta_\mu}^{\rm g}$ the gap equation has at least two
solutions. Moreover, there exists another dividing line inside this
region separating the parameter regimes where the BCS functional is
minimized by $\alpha=0$ or $\alpha\neq 0$, i.e., the normal and the
superfluid regime.

\subsubsection{A toy model in one dimension}

For the purpose of illustration let us now consider the BCS functional
for a one-dimensional system, where the particles interact via a
contact potential of the form $V(x)=-g \delta(x)$ with $g>0$.  In this
case the gap equation takes a very simple form, and the order
parameter $\Delta$ is constant. In fact, the BCS gap equation in this
case is
\begin{equation}\label{1d1}
\frac{1}{g}=\frac{1}{2\pi}\int_\R{\frac{1}{K^\Delta_{T, \delta_\mu} (p)}}dp\,.
\end{equation}
The critical temperatures $T^{\rm i}_{\delta_\mu}$
and $T^{\rm g}_{\delta_\mu}$ are given implicitly via the 
equations
\begin{equation}\label{tl1}
T^{\rm i}_{\delta_\mu}= \left\{(\delta_\mu,T):
\frac1g=\frac{1}{2\pi}\int_\R{\frac{1}{K^0_{T,\delta_\mu}(p)}}dp\right\}
\end{equation}
\begin{equation}\label{tu1}
T^{\rm g}_{\delta_\mu}= \left\{(\delta_\mu,T):\frac{1}{g}=\max_{\Delta}\frac{1}{2\pi}\int_\R{
\frac{1}{K^\Delta_{ T, \delta_\mu} (p)}}dp \right\}.
\end{equation}

The curve $T^{\rm i}_{\delta_\mu}$ encloses the region where the gap
equation \eqref{1d1} has exactly one solution. Above
$T^{\rm o}_{\delta_\mu}$ there are no solutions, while in-between 
\eqref{1d1} has exactly two solutions.  Numerically this is easy to
see, in fact. If one plots $ \frac{1}{2\pi}\int_\R{
  \frac{1}{K^\Delta_{ T, \delta_\mu} (p)}}dp$ as a function of
$\Delta$, one observes that the graph crosses $1/g$ either once, twice, or not at all.

\section{Preliminaries and proof of Theorem \ref{over1}}\label{over1sec}

We start by specifying the precise domain of definition of the BCS functions $\F_T$ in (\ref{freeenergy0}). 

\begin{definition}\label{defF}
Let $\de$ denote the set of functions $(\gamma_+,\gamma_-,\alpha)$,
with $\gamma_+,\gamma_-\in L^1(\R^3,(1+p^2)dp)$, $0\leq
\gamma_+(p),\gamma_-(p)\leq 1$, and $\alpha\in H^1(\R^3,dx)$,
satisfying $|\hat\alpha(p)|^2\leq \gamma_+(p)(1-\gamma_- (p))$ and
$|\hat\alpha(p)|^2\leq \gamma_-(p)(1-\gamma_+(p))$. 
\end{definition}

For $(\gamma_+,\gamma_-,\alpha)\in \de$, the corresponding $\Gamma(p)$
in (\ref{def:Gamma}) satisfies $0\leq \Gamma\leq 1$ as an operator on $\C^4$
for all $p\in \R^3$. Moreover, all the terms in $\F_T$ are well-defined under our assumptions on $V$.

Recall that $\bar \mu = (\mu_+ + \mu_-)/2$ denotes the average
chemical potential and $\delta_\mu = (\mu_+ - \mu_-)/2$. It is easy to see that in the absence of a
potential $V$, $\F_T$ is minimized on $\de$ by the {\em normal
state} given by $\al= 0$, 
$\gamma_+^0(p)=[e^{\frac{1}{T}[(p^2-\bar\mu)
-\delta_\mu]}+1]^{-1}$ and
$\gamma_-^0(p)=[e^{\frac{1}{T}[(p^2-\bar\mu)
+\delta_\mu]}+1]^{-1}$.

We start with the observation that $\F_T$ attains a minimum on $\de$.

\begin{Proposition}
There exists a  minimizer of $\F_T$ in $\de$.
\end{Proposition}

The proof is analogous to \cite[Proposition 2]{HHSS}, and we skip it for simplicity. 

Let us now come to the proof of our main Theorem~\ref{over1}.
We start with some preliminaries. Diagonalizing $\Gamma$ allows to rewrite
the entropy (\ref{def:entropy})  in the  form
\begin{multline} S(\Gamma)=
-\int_{\R^3}[(r+w)\ln(r+w)+(r-w)\ln(r-w)\\
+(s+w)\ln(s+w)+(s-w)\ln(s-w)]dp,
\end{multline}
where the functions $s$, $r$ and $w$ are given by
$$r=\frac{1}{2}(1+\gm_{+}-\gm_{-}), \,\, s=\frac{1}{2}(1-\gm_{+}+\gm_{-}), \,\, w=\frac{1}{2}\sqrt{(1-\gm_{+}-\gm_{-})^2+4|\hat\alpha|^2}.$$ 

\begin{Lemma}\label{l:gapequation}
  Let $(\gm_+, \gm_-,\al)\in \de$ be a minimizer of $\F_T$ for $0<T<\infty$. Then, for a.e. $p\in\R^3$,
\begin{equation}\label{freeenergy4}
(\hat V\ast \hat\al)(p)= -\frac{T}{4}
\hat\alpha(p)\left(f_r{(w)}+f_s{(w)}\right)
\end{equation}
\begin{equation}\label{freeenergy7}
\delta_\mu=\frac{\mu_{+}-\mu_{-}}{2}= \frac{T}{2} \ln
\left(\frac{r^2-w^2}{s^2-w^2}\right),
\end{equation}
\begin{equation}\label{freeenergy8}
(p^2 - \bar\mu)= \frac{T}{4} (1-\gm_+ - \gm_- )(f_r (w) +
f_s(w)),
\end{equation}
where $ f_a(b) =
\frac{1}{b}\ln \frac{a+b}{a-b}$.
\end{Lemma}

\begin{proof}
  Let
  $A_\epsilon:=\{p\in\R^3:(|\hat\al(p)|^2-\gm_+(1-\gm_-))(|\hat\al(p)|^2-\gm_-(1-\gm_+))
  \geq \epsilon\}$ for some $\epsilon >0$. Let $g\in H^1(\R^3)\cap L^1(\R^3)$ be such that
  the Fourier transform  $\hat g$ of $g$ is supported in $A_\epsilon$. Then for small enough $|t|$,
  $(\gamma_+,\gm_-,\alpha+t g)\in \de$. By the Lebesgue
  dominated convergence theorem, it follows that

\begin{align}\label{freeenergy1}\nonumber
\frac{d}{dt}\F_T(\gm_{+},\gm_{-},\al+tg)_{t=0}=& 2\ {\rm \Re\,
}\int\alpha(x) \overline {g(x)}V(x)dx  \\ & + \frac{T}{2} {\rm \Re\,
}\int{\hat\al(p) \overline{\hat g(p)}(f_r{(w)}+f_s{(w)})}dp\,.
\end{align}
Here $\Re$ denotes the real part. Also, for $\hat g$ real,
$(\gamma_{+}+t\hat g,\gm_-,\alpha)\in\de$ for small $|t|$, and we
get
\begin{multline}\label{freeenergy3}
\frac{d}{dt}\F_T(\gm_{+}+t \hat g,\gm_{-},\al)_{t=0}=
\frac{1}{2}\int (p^2-\mu_{ +})\hat g(p)dp   \\ +\frac{T}{2}\int
\left[\frac{(\gm_{+}+\gm_{-}-1)}{2}(f_r{(w)}+f_s{(w)})+\ln
\left(\frac{r^2-w^2}{s^2-w^2}\right)\right]dp
\end{multline}
and, similarly,
\begin{multline}\label{freeenergy2}
\frac{d}{dt}\F_T(\gm_{+},\gm_{-}+t \hat g,\al)_{t=0}=
\frac{1}{2}\int (p^2-\mu_{-})\hat g(p)dp  \\ +\frac{T}{2}\int
\left[\frac{(\gm_{+}+\gm_{-}-1)}{2}(f_r{(w)}+f_s{(w)})-\ln
\left(\frac{r^2-w^2}{s^2-w^2}\right)\right]dp\,.
\end{multline}
Since $(\gm_+,\gm_-,\al)$ minimize $\F_T$ by assumption, the
expressions in (\ref{freeenergy1}),(\ref{freeenergy3}) and
(\ref{freeenergy2}) vanish. This implies that (\ref{freeenergy4}) holds,
and also that 
\begin{equation}\label{freeenergy5}
(p^2-\mu_{-})=\frac{T}{2}\left[\frac{(1-\gm_{+}-\gm_{-})}{2}
(f_r{(w)}+f_s{(w)})+\ln \left(\frac{r^2-w^2}{s^2-w^2} \right)\right]
\end{equation}
\begin{equation}\label{freeenergy6}
(p^2-\mu_{+})=\frac{T}{2}\left[\frac{(1-\gm_{+}-\gm_{-})}{2}(f_r{(w)}+f_s{(w)})-\ln
\left(\frac{r^2-w^2}{s^2-w^2}\right)\right]
\end{equation}
for a.e. $p\in A_\epsilon$. As in \cite{HHSS} we can argue that
the measure of $B:=\R^3 \setminus \cup_{\epsilon>0} A_\epsilon$ is
zero. Subtracting and adding Eqs.~(\ref{freeenergy5}) and
(\ref{freeenergy6}) implies \eqref{freeenergy7} and
\eqref{freeenergy8}.
\end{proof}

\begin{proof}[Proof of Theorem~\ref{over1}]

  We consider the case where $T>0$ in the proof, and leave the analogous case $T=0$
  to the reader. We proceed similarly to \cite[Theorem 1]{HHSS}.
  To show (i), note that
\begin{equation}\label{tempoperator}
K^{0}_{T, \delta_\mu}(p) = \frac{2(p^2 - \bar\mu)}{\tanh
\left(\frac{p^2 - \bar\mu+\delta_\mu}{2T} \right)+\tanh
\left(\frac{p^2 - \bar\mu-\delta_\mu}{2T} \right)} =
\frac{p^2 -\bar\mu}{1-\gamma_+^0 -\gamma_-^0}\,,
\end{equation}
where $\gamma_+^0$ and $\gamma_-^0$ are the momentum distributions in
the normal state defined in the beginning of this section.  We shall
show that if the normal state $\alpha=0$ minimizes $\F_T$, then
$\infspec(K^{0}_{T, \delta_\mu} +\lambda V)\geq 0$.

Since it is true that for every $g \in C_0^\infty$, $||
(1-\gamma_+^0 - \gamma_-^0)^2 +(tg)^2||_\infty <1$ for small enough
$|t|$, it follows from the Lebesgue dominated convergence theorem
that $\frac{d^2}{dt^2}\F_T(\gamma_+^0,\gamma_-^0,t\hat g)$ exists
for small $t$, and equals
\begin{align*}
&\left.\frac{d^2}{dt^2}\F_T(\gm_{+}^0,\gm_{-}^0,t \hat g)\right|_{t=0} \\ &= 2 \lambda \int |\hat g(x)|^2
V(x)dx +T \int {|g(p)|^2 \frac{\left
[ \ln{(\frac{r^0 + w^0}{r^0 - w^0})} + \ln{(\frac{s^0 + w^0}{s^0 -
w^0})}  \right]}{(1-\gamma_+^0 - \gamma_-^0)}}dp
\end{align*}
 at $t=0$. 
Here
$r^0=\frac{1}{2}(1+\gm_{+}^0-\gm_{-}^0)$,
$s^0=\frac{1}{2}(1-\gm_{+}^0+\gm_{-}^0)$ and
$w^0=\frac{1}{2}(1-\gm_{+}^0-\gm_{-}^0)\,.$ It is easy to see that
$$\ln{\left(\frac{r^0 + w^0}{r^0 - w^0}\right)}+\ln{\left(\frac{s^0 + w^0}{s^0 -
w^0}\right)}= \frac{2(p^2-\bar\mu)}{T}\,,$$ and hence
$$\frac{d^2}{dt^2}\F_T(\gm_{+}^0,\gm_{-}^0,t \hat g)_{t=0}= 2 \lambda \int |\hat g(x)|^2
V(x)dx +2\int K^0_{T,\delta_\mu}(p) |g(p)|^2 dp\,.$$ If the normal
state minimizes $\F_T$ then clearly
$\frac{d^2}{dt^2}\F_T(\gm_{+}^0,\gm_{-}^0,t \hat g)_{t=0}\geq 0$
(keeping in mind that it is also stationary).  This implies that
$\bra \hat g,(K^{0}_{T, \delta_\mu}+\lambda V) \hat g\ket \geq 0$ for
all $g \in C_0^\infty$, and hence proves the claim.

Next we shall show (ii). For a minimizer $(\gamma_+,
\gamma_-, \al)$ of $\F_{T}$, we can define $\Delta$ via
\begin{equation}\label{relda}
\Delta (p) =
2\dfrac{p^2-\bar\mu}{1-\gamma_+-\gamma_-}\hat\alpha(p)\,.
\end{equation}
Setting further $E(p) = \sqrt{(p^2 - \bar \mu)^2 +
|\Delta(p)|^2}$ and recalling that
$$ 2w = \sqrt{(1-\gm_+ - \gm_- )^2 +4|\hat\alpha|^2} $$
we obtain that
$$
E(p)  = 2w(p) \left| 
\frac{p^2-\bar\mu}{1-\gm_+(p) - \gm_-(p) } \right|\,.
$$
If we further use (\ref{freeenergy8}) we get 
\begin{equation}\label{ae1}
E(p) = w\frac{T}{2}(f_r(w)+f_s(w)).
\end{equation}
It then follows from the definition of $\Delta(p)$ and $E(p)$ that
\begin{equation}\label{alpha1}
\frac{\Delta(p)}{E(p)} w(p) = \hat\al(p).
\end{equation}
Noting that
$$ (f_r(w) + f_s(w))=\frac{1}{\omega}\ln{\frac{(r+w)(s+w)}{(r-w)(s-w)}}, $$
we obtain from (\ref{freeenergy7}) and the definition of $E(p)$ that
$$\frac{E(p) + \delta_\mu}{T} = \ln \left({\frac{r+w}{s-w}}\right),
\quad \frac{E(p) - \delta_\mu}{T} = \ln \left({\frac{s+w}{r-w}}\right).$$
Since it is also true that
\begin{equation}\label{w}
1 - \left(1+\frac{r+w}{s-w} \right)^{-1} -
\left(1+\frac{s+w}{r-w}\right)^{-1}=2w,
\end{equation}
we get from (\ref{freeenergy4}) and (\ref{alpha1}) that
\begin{equation}\label{vad}
-\lambda V * \hat\al=\frac{\Delta(p)}{2},
\end{equation}
and hence
$$
\Delta = -2\lambda V * \hat\al = -\lambda V*\frac{\Delta}{E}
\left(1 - \frac{1}{1+\frac{r+w}{s-w}} -
\frac{1}{1+\frac{s+w}{r-w}}\right).
$$
From the above expressions, we observe that
$$e^{\frac{E(p) + \delta_\mu}{T}} = \frac{r+w}{s-w}, \quad e^{\frac{E(p) - \delta_\mu}{T}} = \frac{s+w}{r-w}.$$
Therefore we arrive at the BCS gap equation
\begin{equation}
\Delta = -\lambda
V*\frac{\Delta}{E}\left(1-\left(1+e^{\frac{E(p)+\delta_\mu}{T}}
\right)^{-1} -\left(1+e^{\frac{E(p)-\delta_\mu}{T}}
\right)^{-1}\right)\,.
\end{equation}

To see (iii), note that the gap equation can be written
as $$(K^{\Delta}_{T, \delta_\mu} +\lambda V) \alpha = 0\,,$$ with
$\hat\alpha(p)=\Delta(p)/K^{\Delta(p)}_{T,\delta_\mu}(p)$.  Consequently,
\begin{equation}
\bra \al,(\widetilde K _{T, \delta_\mu}+\lambda V) \al\ket = \bra
\al,(\widetilde K _{T, \delta_\mu}-K^{\Delta}_{T, \delta_\mu})
\al\ket \leq 0
\end{equation}
by the definition of $ \widetilde K _{T, \delta_\mu}$. The inequality
is, in fact, strict if $\Delta$ is not identically zero, since $\hat
\al$ and $\Delta$ have the same support. This completes the proof of
Theorem~\ref{over1}.
\end{proof}

\begin{corollary}\label{cor1}
Assume that $\cosh(\delta_\mu/T) \leq 2$. Then the existence of a non-trivial solution of the BCS gap equation is {\em equivalent} to $K_{T,\delta_\mu}^0 + \lambda V$ having a negative eigenvalue.
\end{corollary}

\begin{proof}
  The corollary follows directly from Theorem~\ref{over1} and the fact that $K^{\Delta}_{T,
    \delta_\mu}(p)$ is pointwise monotone in $\Delta$ for all
  $\delta_\mu/T \in [0,\cosh^{-1}(2)]$. Hence $K^0_{T,\delta_\mu}
  = \widetilde K_{T,\delta_\mu}$ in this case. To see this one checks that  the function $x \mapsto (\tanh (x+c) +
    \tanh(x-c))/x$ is monotone decreasing on $\R_+$ if
  $\cosh(2c) \leq 2$.
\end{proof}

\section{Bounds on the critical temperature}\label{sec:temp}

As explained in Subsect.~\ref{ss:tc}, the two linear operators in
Theorem~\ref{over1} determine upper and lower bounds $T^{\rm o}_{\delta_\mu}$
and $T^{\rm i}_{\delta_\mu}$ on the critical temperature. In
this section we evaluate $T^{\rm i}_{\delta_\mu}(\lambda V)$ and
$T^{\rm o}_{\delta_\mu}(\lambda V)$ in the limit of small $\lambda$. We will
proceed similarly to \cite{FHNS,HS,HS2}. Note that it follows from
Corollary~\ref{cor1} that for $\cosh(\delta_\mu/T)\leq 2$ the two
curves coincide, i.e., $T^{\rm i}_{\delta_\mu}(\lambda
V)=T^{\rm o}_{\delta_\mu}(\lambda V)$. (Compare with Fig.~\ref{fig1}.)

We start with $T^{\rm i}_{\delta_\mu}(\lambda V)$, defined in
(\ref{def:ti}).  The Birman-Schwinger principle implies that a pair
$(\delta_{\mu},T) \in T^{\rm i}_{\delta_\mu}(\lambda V)$ is
characterized by the fact that the compact operator
\begin{equation}\label{def:B}
B_{T,\delta_\mu} = \lambda |V|^{\frac12}\frac 1{K^{0}_{ T, \delta_{\mu}}}V^{\frac12}
\end{equation}
has $-1$ as its lowest eigenvalue. Here we use the
notation $$V(x)^{1/2} = ({\rm sgn\,} V(x)) |V(x)|^{1/2}.$$ That $-1$
has to be the lowest eigenvalue of $B_{T,\delta_\mu}$ follows from the
monotonicity of $K^0_{T,\delta_\mu}$ in $\delta_\mu$, similarly to
\cite[Lemma 3.1]{FHNS}.  In fact, since the spectrum of
$B_{T,\delta_\mu}$ goes to $0$ for large $\delta_\mu$ and is
continuous in $\delta_\mu$, an eigenvalue of $B_{T,\delta_\mu}$
smaller than $-1$ would correspond to an eigenvalue $-1$ for larger
$\delta_\mu$, which is in contradiction to the monotonicity of
$K^0_{T,\delta_\mu} + \lambda V$ in $\delta_\mu$ and $0$ being the lowest
eigenvalue for the critical parameter.

For $\bar\mu >0$, the operator $B_{T,\delta_\mu}$ becomes singular as $(\delta_\mu, T) 
\to (0,0)$, and the key observation is that its singular part is
represented by an operator $\V_{\bar\mu}: \,
L^2(\Omega_{\bar\mu}) \to L^2(\Omega_{\bar\mu})$, given by 
\begin{equation}\label{defvm}
 \big(\V_{\bar\mu} u\big)(p) =
  \frac 1{(2\pi)^{3/2}} \frac 1{\sqrt{\bar\mu}}\int_{\Omega_{\bar\mu}}\hat V(p-q) u(q) \,d\omega(q)\,,
\end{equation}
where $d\omega$ is the uniform Lebesgue measure on
$\Omega_{\bar\mu}$. We note that the operator $\V_\mu$ has
already appeared in the literature \cite{BY,LSW}. If we assume that
$V\in L^1(\R^3)$, then $\hat V(p)$ is a bounded continuous
function, and hence $\V_{\bar\mu}$ is a Hilbert-Schmidt
operator. In fact, $\V_{\bar\mu}$ is trace class, and its trace
equals $\frac{\sqrt{\bar\mu}}{2\pi^2} \int_{\R^3} V(x)dx$. Let
\begin{equation}\label{e-1}
e_{\bar\mu}= \inf {\rm spec} \,\V_{\bar\mu}
\end{equation}
denote the infimum of the spectrum of $\V_{\bar\mu}$.  Since
$\V_{\bar\mu}$ is compact, we have $e_{\bar\mu}\leq 0$.

Let $\F: L^1(\R^3) \to L^2(\Omega_{\bar\mu})$ denote the
bounded operator which maps $\psi\in L^1(\R^3)$ to the Fourier
transform of $\psi$, restricted to the Fermi sphere
$\Omega_{\bar\mu}$. For $V\in L^1(\R^3)$, multiplication by
$|V|^{1/2}$ is a bounded operator from $L^2(\R^3)$ to $L^1(\R^3)$,
and therefore $\F |V|^{1/2}$ is a bounded operator from $L^2(\R^3)$
to $L^2(\Omega_{\bar\mu})$.

Finally, as in \cite{HS}, we define $\W_{\bar\mu}$ via its
quadratic form, to be an operator on $L^2(\Omega_{\bar\mu})$ such that
\begin{align}\nonumber
  \langle u | \W_{\bar\mu} |u \rangle = \int_{0}^\infty d|p| & \left( \frac
    {|p|^2}{\big||p|^2-\bar\mu\big|} \left[ \int_{\S^2} d\Omega \left(
        |\hat\varphi(p)|^2 -
        |\hat\varphi(\sqrt{\bar\mu} p/|p|)|^2 \right)\right] \right. \\
  \label{defW} & \quad \left. + \frac 1{|p|^2} \int_{\S^2} d\Omega\,
    |\hat\varphi(\sqrt{\bar\mu} p/|p|)|^2\right) \,.
\end{align}
Here, $(|p|,\Omega)\in \R_+\times \S^2$ denote spherical coordinates
for $p\in\R^3$ and $\hat\varphi(p) = (2\pi)^{-3/2} \int_{\Omega_\mu}
\hat V(p-q) u(q) d\omega(q)$. It was shown in \cite{FHNS} that
(\ref{defW}) is well defined, despite the apparent singularity of the integral, and that $\W_{\bar\mu}$ is of
Hilbert-Schmidt class.

For $\lambda >0$, let
\begin{equation}\label{defB}
  \B_{\bar\mu} = \lambda \frac \pi {2\sqrt {\bar\mu}} \V_{\bar\mu}
- \lambda^2 \frac{\pi}{2{\bar\mu}} \W_{\bar\mu}\,,
\end{equation}
and let $\varrho(\lambda)$ denote its ground state energy,
\begin{equation}\label{defbm}
  \varrho(\lambda) = \inf {\rm spec} \, \B_{\bar\mu} \,.
\end{equation}
We note that if $e_{\bar\mu}<0$, then $\varrho(\lambda)< 0$ for
small $\lambda$.  If the (normalized) eigenfunction $u \in
L^{2}(\Omega_{\bar\mu})$ corresponding to $e_{\bar\mu}$ is
unique, then
\begin{equation}\label{ddeff}
\varrho(\lambda)=\bra u | \B_{\bar\mu} | u \ket +o(\lambda^3).
\end{equation}
In the degenerate case, this formula holds if one chooses $u$ to be
the eigenfunction of $\V_\mu$ that yields the largest value $\langle
u|\W_\mu|u\rangle$ among all such (normalized) eigenfunctions.

\subsection{Evaluation of $T^{\rm i}_{\delta_\mu}$.}

In this subsection we will derive the asymptotic behavior
\eqref{claim:tc} of $T^{\rm i}_{\delta_\mu}(\lambda V)$ for small
$\lambda$, which we present in the following theorem.

\begin{theorem}\label{lowerb}
  Let $V \in L^1 (\R^3) \cap L^{3/2}(\R^3)$ be real-valued and
  $\bar\mu>0$.  Assume that $e_{\bar\mu}$ defined in
  (\ref{e-1}) is strictly negative, and let $\varrho(\lambda)$ be
  defined as in (\ref{defbm}). Then any sequence of pairs
  $\left(\delta_\mu(\lambda),T(\lambda)\right)$ on the curve
  $T_{\delta_\mu}^{i}(\lambda V)$ satisfies
\begin{equation}\label{7-1}
\lim_{\lambda \to 0}\left( \ln
{\frac{\bar\mu}{T}} + \frac{\pi}{2
\sqrt{{\bar\mu}}\varrho(\lambda)} -
\kappa^{\rm i}(\delta_\mu/T)\right) = 2 - \gamma - \ln (8/\pi) 
\end{equation}
where  $\gamma\approx 0.5772$ is Euler's constant and 
\begin{equation}\label{kappai}\kappa^{\rm i}(t)=\frac{1}{1+e^{t}}\int_{0}^{\infty}
\frac{1-e^{-x}}{1+e^{x+t}}\,\frac{dx}x +
\frac{1}{1+e^{-t}}\int_{0}^{\infty}
\frac{1-e^{-x}}{1+e^{x-t}}\, \frac{dx}x - \ln \frac\pi 2\,.
\end{equation}
\end{theorem}
Note that $\kappa^{\rm i}(t) = \ln(t)+\gamma - \ln(\pi/2) + o(1)$ for large $t$, while $\kappa^{\rm i}(0)=0$.
We can rewrite Eq.~(\ref{7-1}) in the form
\begin{align}\nonumber
\kappa^{\rm i} \left(\frac{\delta_\mu}{{T}}\right) 
=\ln{\left( \frac{8 \bar\mu e^{\gamma -2}e^{\frac{\pi}{2
\sqrt{{\bar\mu}}\varrho(\lambda)}
}}{\pi T}\right)} + o(1).
 \end{align}
 If $T_c$ denotes the critical temperature at $\delta_\mu = 0$,
 defined in Subsect.~\ref{sec:bal}, this can also be written as
\begin{equation}
\kappa^{\rm i} \left(\frac{\delta_\mu}{T}
\right)=\ln \left(\frac{T_c}{T} \right) + o(1)\,,
\end{equation}
which agrees with (\ref{claim:tc}) and is plotted in Figure
\ref{fig1}. The region denoted by I defines the region where the
system is guaranteed to be in a superfluid phase. The asymptotic
behavior of $\kappa^{\rm i}$ implies that the curve hits the vertical
axis at $1$ and the horizontal axis at $\pi e^{-\gamma}/2\approx
0.88$.

\begin{proof}Let
\begin{equation}\label{m0}
 m(T,\delta_\mu) =  \frac 1{4\pi \bar\mu}
\int_{\R^3} \left( \frac {1}{K^{0}_{ T, \delta_\mu}(p)}
  - \frac 1{p^2}\right) dp \,,
\end{equation}
and let $M(T,\delta_\mu)$ be the operator 
\begin{equation}\label{bigM}
M(T,\delta_\mu) = \frac{1}{K^{0}_{ T, \delta_\mu}}
-m(T,\delta_\mu) \F \F^*\,,
\end{equation}
with $\F$ defined in the paragraph after Eq.~(\ref{e-1}).  Following
\cite[Lemma 2]{FHNS} it is not difficult to see that
$V^{1/2}M(T,\delta_\mu)|V|^{1/2}$ is bounded in Hilbert-Schmidt
norm uniformly in $(T,\delta_\mu)$.

Suppose $\psi $ is an eigenstate of the
operator $K^{0}_{ T, \delta_\mu} + \lambda V$ corresponding to the lowest eigenvalue $0$.
By the
Birman-Schwinger principle, $\phi = V^{\frac12}\psi$ satisfies 
$$ -\phi = B_{T,\delta_\mu} \phi\,, $$
with $B_{T,\delta_\mu}$ defined in (\ref{def:B}).  Moreover, $-1$ is
the lowest eigenvalue of $B_{T,\delta_\mu}$, as argued above.

Since $V^{1/2}M(T,\delta_\mu)|V|^{1/2}$ is bounded uniformly in
$(T,\delta_\mu)$, the operator $1+\lambda V^{1/2} M(T,\delta_\mu) |V|^{1/2}$ is
invertible for $\lambda$ small enough.  Therefore
\begin{align}\label{nbs1}
1+ \lambda |V|^{\frac12}\frac{1}{{K^{0}_{ T, \delta_\mu}}}V^{\frac12}&=
1+\lambda |V|^{\frac12} \left(m(T,\delta_\mu) \F \F^* +
M(T,\delta_\mu)\right)V^{\frac12} \\ \nonumber &=D\left(1+ \lambda
m(T,\delta_\mu)D^{-1} V^{\frac12} \F^* \F |V|^{\frac12}\right),
\end{align}
where $D=1+\lambda V^{\frac12} M(T,\delta_\mu) |V|^{\frac12}$. The
fact that $-1$ is an eigenvalue of (\ref{def:B}) yields
\begin{equation}\label{nbs2}
D \left(1+
\lambda m(T,\delta_\mu)\frac{1}{1+\lambda V^{\frac12}
M(T,\delta_\mu) |V|^{\frac12}} V^{\frac12} \F^* \F
|V|^{\frac12}\right)\phi=0.
\end{equation}
Since $D$ is
invertible, this is equivalent to the operator
\begin{equation}\label{nbs4}
\lambda m(T,\delta_\mu) \F |V|^{1/2} \frac{1}{1+\lambda V^{1/2}
M(T,\delta_\mu) |V|^{1/2}} V^{1/2} \F^*
\end{equation}
having an
eigenvalue $-1$ (using that $AB$ and $BA$ are isospectral if they
are compact). Note that the latter operator is acting on
$L^2(\Omega_{\bar\mu})$.

By expanding $(1+\lambda V^{\frac12} M(T,\delta_\mu)
|V|^{\frac12})^{-1}$ in a Neumann series we obtain the following
implicit characterization of the corresponding temperature:
\begin{equation}\label{nbs5}
\lambda m(T,\delta_\mu) \infspec (\F \left(V - \lambda V
M(T,\delta_\mu) V + o(\lambda^2) \right) \F^* )= - 1.
\end{equation}
Note that $\F V \F^* = \sqrt{\bar\mu} \,\V_{\bar\mu}$, which
was defined in (\ref{defvm}). It follows that to leading order
\begin{equation}\label{nbs6}
\lim_{\lambda \to 0} \lambda m(T,\delta_\mu) = - \frac{1}{\inf {\rm
spec}  \F V \F^*} = - \frac 1 {\sqrt{\bar\mu} \, e_{\bar\mu}} .
\end{equation}
To derive the second order correction we employ $\W _{\bar\mu}$,
which was defined in (\ref{defW}). It follows from (\ref{nbs5}), first
order perturbation theory, the fact that $\F V \F^*$ is compact and
that $\infspec \F V \F^*<0$ by assumption, that
\begin{equation}\label{deno}
 m(T,\delta_\mu) = \frac {-1}{ \lambda \langle u| \F V \F^*| u\rangle -
    \lambda ^2 \langle u| \F V M(T,\delta_\mu) V \F^*| u\rangle + O(\lambda^3)}\,,
\end{equation}
where $u$ is the normalized eigenfunction corresponding to the
lowest eigenvalue of $\F V \F^*$.  If $u$ is degenerate, we choose
the $u$ that minimizes the $\lambda^2$ term in the denominator of
(\ref{deno}) among all such eigenfunctions. Eq.~(\ref{deno})
represents an implicit equation for $T_{\delta_\mu}^{i}$. Since $\F
V M(T,\delta_\mu) V\F^*$ is uniformly bounded and $T \to 0$ as
$\lambda \to 0$, we have to evaluate the limit of $\langle u| \F V
M(T,\delta_\mu) V \F^*| u\rangle$ as $(\delta_\mu,T)  \to 0$. For
this purpose, let $\varphi = V \F^* u$. Then
\begin{align}\nonumber
  &\langle u| \F V M(T,\delta_\mu) V \F^* |u\rangle \\ \label{comeq}
& = \int_{\R^3} \frac 1{K^{0}_{ T, \delta_\mu}(p)}
  |\hat\varphi(p)|^2 \, dp -m(T,\delta_\mu) \int_{\Omega_{\bar\mu}}
  |\hat\varphi(p)|^2 \, d\omega(p) \\ \nonumber & = \int_{\R^3}
  \left(\frac 1{K^{0}_{ T, \delta_\mu}(p)} \left[ |\hat\varphi(p)|^2
      -|\hat\varphi(\sqrt{\bar\mu} p/|p|)|^2 \right] + \frac 1{p^2}
    |\hat\varphi(\sqrt{\bar\mu} p/|p|)|^2 \right) dp\,.
\end{align}
Recall the definition of $K^{0}_{T,\delta_\mu}(p)$ in (\ref{def:kD}).
For a fixed $\delta_\mu/{T}$, it converges to $|p^2-\bar\mu|$ as
$T\to 0$.  Using the fact that the spherical average of
$|\hat\varphi(p)|^2$ is Lipschitz continuous (see \cite{HS}), we can
interchange the limit and the radial integral over $|p|$, and hence
obtain
\begin{equation}\label{deno2}
  \lim_{T\to 0} \langle u| \F V M(T,\delta_\mu) V \F^*| u\rangle  = \langle u| \W_{\bar\mu} |u \rangle\,.
\end{equation}
We have thus shown that  for $(\delta_\mu,T) \in T^{\rm i}_{\delta_\mu}$
\begin{equation}\label{deno3}
  \lim_{\lambda\to 0} \left(m(T,\delta_\mu) + \frac 1{\inf {\rm spec} \left
        (\lambda \sqrt {\bar\mu}\, V_{\bar\mu} - \lambda^2 \W_{\bar\mu}\right)} \right) =
  0\,.
\end{equation}

It remains to compute $m(T,\delta_\mu)$.

\begin{Lemma}\label{calcm}
  As $(T, \delta_\mu) \to (0,0)$,
\begin{eqnarray}\nonumber
    m(T,\delta_\mu) &=& \frac 1{4\pi \bar\mu} \int_{\R^3}\left( \frac 1
      {K^{0}_{ T, \delta_\mu}(p) } -\frac 1{p^2}\right) dp  \\ \label{mo} &=& \frac 1{\sqrt{\bar\mu}}
\left( \ln \frac {\bar\mu}{T} + \gamma - 2 +\ln \frac 8\pi
-\kappa^{\rm i}(\delta_\mu/T) + o(1)\right).
\end{eqnarray}
\end{Lemma}

 \begin{proof}
  We split the integral into two parts according to whether $p^2\leq \bar\mu$ or 
  $p^2\geq \bar\mu$. By changing variables from $p^2-\bar\mu$ to $-t$ and
  $t$, respectively, we see that $ m(T,\delta_\mu)$ equals 
 \begin{multline}\label{lowerblemma}
    \frac{1}{2\bar\mu}\int_{0}^{\bar\mu}\left( \left({1-\frac{1}{1+e^{\frac{1}{T}(t+\delta_\mu )}}-\frac{1}{1+e^{\frac{1}{T}(t-\delta_\mu)}}}\right)
  \frac{\sqrt{\bar\mu - t}}{t}-\frac{1}{\sqrt{\bar\mu - t}}\right)dt   \\  +
  \frac{1}{2\bar\mu}\int_{0}^{\infty}\left( \left({1-\frac{1}{1+e^{\frac{1}{T}(t+\delta_\mu)}}-\frac{1}{1+e^{\frac{1}{T}(t-\delta_\mu)}}}\right)
  \frac{\sqrt{\bar\mu + t}}{t}-\frac{1}{\sqrt{\bar\mu + t}}\right)dt\,.
\end{multline}
 To simplify the notation, let us introduce the function
 \begin{equation}\label{def:ups}
  \Upsilon_0(t)= \left({1-\frac{1}{1+e^{\frac{1}{T}(t+\delta_\mu)}}-\frac{1}{1+e^{\frac{1}{T}(t-\delta_\mu)}}}\right)\,.
  \end{equation}
 We can rewrite (\ref{lowerblemma}) as
  \begin{align}\nonumber 
  m(T,\delta_\mu) & =   \frac{1}{2\bar\mu}\int_{\bar\mu}^{\infty}\left( \Upsilon_0(t)
  \frac{\sqrt{\bar\mu + t}}{t}-\frac{1}{\sqrt{\bar\mu + t}}\right)dt    \\ \nonumber & \quad +
  \frac{1}{2\bar\mu}\int_{0}^{\bar\mu}  \Upsilon_0(t)\left(
  \frac{\sqrt{\bar\mu + t} +\sqrt{\bar\mu - t} - 2\sqrt{\bar\mu}}{t}\right)dt
    \\  & \quad -
  \frac{1}{2\bar\mu}\int_{0}^{\bar\mu}\left(\frac{1}{\sqrt{\bar\mu + t}} + \frac{1}{\sqrt{\bar\mu - t}}\right)dt   +
  \frac{1}{\sqrt{\bar\mu}}\int_{0}^{\bar\mu}
  \frac{ \Upsilon_0(t)}{t}dt\,. \label{7-11}
  \end{align}
Using the dominated convergence theorem, it is easy to see that
  \begin{align}\nonumber
    &\lim_{T, \delta_\mu \to 0}\int_{\bar\mu}^{\infty}\left( \Upsilon_0(t)
  \frac{\sqrt{\bar\mu + t}}{t}-\frac{1}{\sqrt{\bar\mu + t}}\right)dt \\ \nonumber &=
    \int_{\bar\mu}^{\infty} \left(\frac{\sqrt{\bar\mu+t}}{t} - \frac
      1{\sqrt{\bar\mu+t}}\right) dt = 2 \sqrt{\bar\mu} \ln\left(1+\sqrt
      2\right)\,.
  \end{align}
Moreover,  $$\frac{1}{2\bar\mu}\int_{0}^{\bar\mu}\left(\frac{1}{\sqrt{\bar\mu + t}} + \frac{1}{\sqrt{\bar\mu - t}}\right)dt = \sqrt\frac 2{\bar\mu} \,.$$
Again by the dominated convergence theorem,  the second integral becomes
\begin{eqnarray}\nonumber
  \frac{1}{2\bar\mu}\int_{0}^{\bar\mu}\frac{\sqrt{\bar\mu+t} + \sqrt{\bar\mu -t} -2\sqrt{\bar\mu}}{t}dt   =
  \frac{1}{\sqrt{\bar\mu}}\left( \ln 4-\ln(1+\sqrt{2}) + \sqrt{2}  -2
  \right)
 \end{eqnarray}
 in the
limit $T, \delta_\mu \to 0$. 
 With $c={\delta_\mu}/{T}$, 
 using the fact that $[1+e^{c}]^{-1}+[1+e^{-c}]^{-1}=1$,
 we can rewrite the last term in (\ref{7-11}) as
 $$ \frac{1}{\sqrt{\bar\mu}}\left[ \int_{0}^{\frac{\bar\mu}{T}}\frac{1}{x}\left(\frac{1}{1+e^{c}}-\frac{1}{1+e^{c}e^{x}}\right)dx +
 \int_{0}^{\frac{\bar\mu}{T}}\frac{1}{x}\left(\frac{1}{1+e^{-c}}-\frac{1}{1+e^{-c}e^{x}}\right)dx \right].$$
 Note that
 \begin{eqnarray}\nonumber
 \int_{0}^{\frac{\bar\mu}{T}}\frac{1}{x}\left(\frac{1}{1+e^{c}}-\frac{1}{1+e^{c}e^{x}}\right)dx &
 = &
 \frac{1}{1+e^{c}}\int_{0}^{\frac{\bar\mu}{T}}\frac{1-e^{-x}}{x}dx
 {}\nonumber\\ &&{} \nonumber -
 \frac{1}{1+e^{c}}\int_{0}^{\frac{\bar\mu}{T}}\frac{1-e^{-x}}{1+e^{c}e^{x}}\, \frac {dx}x\,.
\end{eqnarray}
Using integration by parts,
\begin{align}\nonumber
   \int_{0}^{\bar\mu/T}
  \frac {1-e^{-x}}{x} dx &= \ln \frac{\bar\mu}{T} \left(
    1-e^{-{\bar\mu}/T}\right) - \int_{0}^{\bar\mu/T} \ln (x)e^{-x} \, dx\,.
\end{align}
Moreover, since 
 \begin{equation}\label{euler}
   -\int_{0}^{\infty} e^{-x} \ln{(x)} dx = \gamma
   \end{equation}
(Euler's constant), we conclude that
$$\lim_{T, \delta_\mu \to 0}\left(  \int_{0}^{\bar\mu/T}
  \frac {1-e^{-x}}{x} dx - \ln \frac{\bar\mu}{T} \right) =
\gamma\,.$$ The same argument applies with $c$ replaced by $-c$, and
hence
$$ \frac{1}{\sqrt{\bar\mu}}\int_{0}^{\bar\mu}
  \frac{ \Upsilon_0(t)}{t}dt = \frac{1}{\sqrt{\bar\mu}}\left( \gamma + \ln{\frac{\bar\mu}{T}} -\ln \frac 8\pi- \kappa^{\rm i}(c)
+ o(1) \right).$$
Combining all the terms gives the statement of the lemma.
\end{proof}

 Theorem \ref{lowerb} follows
from (\ref{deno3}), (\ref{defbm}), and (\ref{mo}).
\end{proof}

\subsection{Evaluation of $T^{\rm o}_{\delta_\mu}$}

For $x\geq 0$ and $c\geq 0$, consider the function 
$$f(x,c)=\frac{x}{\tanh \frac{x+c}2 + \tanh\frac{x-c}2}\,.$$
For each $c$ there is a $b(c)$ such that $x \mapsto f(x,c)$ attains its minimum at $b(c)$ and  $f$
is monotone increasing for all $x \geq b(c)$.  Recall from Corollary \ref{cor1}
that $b(c)=0$ for $c \leq  \cosh^{-1}(2)$ . We have 
$$
K_{T, \delta_\mu}^{\Delta}= T \, f\left(\sqrt{ (p^2 -
\bar\mu)^2 + |\Delta|^2}/T,{\delta_\mu}/{T}\right)\,.
$$
Therefore, $K_{T, \delta_\mu}^{\Delta}$ is monotone in $\Delta$
whenever $\frac{|p^2 - \bar\mu|}{T} \geq
b\left({\delta_\mu}/{T}\right)$, and we have that
\begin{equation}\label{Kd}
\widetilde K _{T, \delta_\mu} = \inf_{\Delta} K_{T,
\delta_\mu}^{\Delta} = \left\{ \begin{array}{ll} K^{0}_{T,
\delta_\mu} & {\rm for\ } \frac{|p^2 - \bar\mu|}{T} \geq
b \left({\delta_\mu}/{T}\right) \\
2 T f(b(\delta_\mu/T),{\delta_\mu}/{T}) & {\rm for\ } \frac{|p^2 -
\bar\mu|}{T} < b \left({\delta_\mu}/{T}\right)\,.
\end{array}\right.
\end{equation}

\begin{theorem}\label{upperb}
  Let $V \in L^1 (\R^3) \cap L^{3/2}(\R^3)$ be real-valued and
  $\bar\mu > 0$.  Assume that $e_{\bar\mu}$ defined in
  (\ref{e-1}) is strictly negative, and let $\varrho(\lambda)$ be
  defined as in (\ref{defbm}).  Then any sequence of pairs
  $\left(\delta_\mu(\lambda),T(\lambda)\right)$ on the curve
  $T_{\delta_\mu}^{\rm o}(\lambda V)$ satisfies
\begin{equation}\label{7o-1}
\lim_{\lambda \to 0}\left( \ln {\frac{\bar\mu}{T}} +
\frac{\pi}{2 \sqrt{{\bar\mu}}\varrho(\lambda)} -   \kappa^{\rm o} (\delta_\mu/T)\right) = 2 -
\gamma - \ln \frac 8\pi \,,
\end{equation}
where $\gamma\approx 0.5772$ is Euler's constant and 
\begin{align}\label{kappao}
\kappa^{\rm o}(c)&= \frac{1}{1+e^{c}}\int_{b}^{\infty}\frac{1-e^{-x}}{1+e^{x+c}}\, \frac {dx}x+\frac{1}{1+e^{-c}}\int_{b}^{\infty}\frac{1-e^{-x}}{1+e^{x-c}}\, \frac{dx}x
\\ \nonumber & \ \
 +\left( 1-e^{-b}\right)\ln{b} -
\int_{0}^{b} \ln (x)e^{-x} \ dx
- \frac{b}{2f(b,c)} - \ln \frac \pi 2 \,,
\end{align}
with $b= b(c)$ defined above.
\end{theorem}

Recall that $b(c)=0$ for $c\leq \cosh^{-1}(2)$, and hence $\kappa^{\rm
  o}(c) = \kappa^{\rm i}(c)$ in this case. For larger $c$ they differ,
however, as Figure~\ref{fig1} shows.

\begin{proof}
The proof works analogously to $T^{\rm i}_{\delta_\mu}$. The only
difference is that $m(T,\delta_\mu)$ now has to be replaced by 
\begin{equation}\label{m0tilde}
 \widetilde m(T,\delta_\mu) = \frac 1{4\pi \mu}
\int_{\R^3} \left( \frac {1}{\widetilde K _{T, \delta_\mu}(p)}
  - \frac 1{p^2}\right) dp\,,
\end{equation}
and the operator $M(T,\delta_\mu)$ by 
\begin{equation}\label{bigm}
\widetilde M(T,\delta_\mu) = \frac{1}{\widetilde K _{T, \delta_\mu}(p)} -
\widetilde m(T,\delta_\mu) \F \F^*\,.
\end{equation}
The result is that
\begin{equation}\label{denof2}
  \lim_{\lambda\to 0} \left( \widetilde m(T,\delta_\mu)  + \frac 1{\inf {\rm spec }\left
        (\lambda \sqrt \mu\, \V_{\bar\mu} - \lambda^2 \W_{\bar\mu}\right)} \right) =
  0\,
\end{equation}
for $(\delta_\mu,T) \in T_{\delta_\mu}^{\rm o}$. What remains is to calculate $\widetilde m(T,\delta_\mu)$.

\begin{Lemma}\label{calcm2}
  As $(T, \delta_\mu) \to (0,0)$,
\begin{equation}\label{lemresult}
    \widetilde m (T,{\delta_\mu})= \frac 1{\sqrt{\bar\mu}}
\left( \ln \frac {\bar\mu} T + \gamma - 2 +\ln \frac 8\pi
-\kappa^{\rm o}(\delta_\mu/T) + o(1)\right)\,.
\end{equation}
\end{Lemma}

 \begin{proof}
   For simplicity we will abbreviate $b \left({\delta_\mu}/{T}\right)$
   by $b$ throughout the proof. First note that from (\ref{Kd}) it
   follows that $\widetilde m (T,{\delta_\mu})$ equals
\begin{equation}\label{k-Delta}
   \frac 1{4\pi \bar\mu}\left[ \int_{\frac{|p^2 -
\bar\mu|}{T} \geq b}\left( \frac 1 {K^{0}_{T, \delta_\mu}(p) }
-\frac 1{p^2}\right) dp  + \int_{\frac{|p^2 - \bar\mu|}{T} <
b}\left( \frac 1
      {f(b,{\delta_\mu}/{T})} -\frac 1{p^2}\right)dp\right] \,.
\end{equation}
For given $b>0$, 
\begin{equation}\nonumber
\lim _{T, \delta_\mu \to 0} \frac 1{4\pi \bar\mu}
\int_{\frac{|p^2 - \bar\mu|}{T} < b}\left( \frac 1
      {2 T f(b,{\delta_\mu}/{T})} -\frac 1{p^2}\right) dp = \frac{b}{2\sqrt{\bar\mu}f(b,\delta_\mu/T)}\,,
\end{equation}
and the result obviously also holds for $b=0$, where both sides are zero. 

As in the proof of Lemma~\ref{calcm}, we split the first integral in (\ref{k-Delta})  into two parts
according to $p^2\leq \bar\mu$
  and $p^2\geq \bar\mu$, and change variables from $p^2-\bar\mu$ to $-t$ and
  $t$, respectively. Introducing again the function $\Upsilon_0$ in (\ref{def:ups}), the integral equals 
  \begin{align}\nonumber
 &  \frac{1}{2\bar\mu}\int_{\bar\mu}^{\infty}\left(  \Upsilon_0(t)
  \frac{\sqrt{\bar\mu + t}}{t}-\frac{1}{\sqrt{\bar\mu + t}}\right)dt \\ \nonumber & +
  \frac{1}{2\bar\mu}\int_{T b}^{\bar\mu}  \Upsilon_0(t)
  \left(\frac{\sqrt{\bar\mu + t} +\sqrt{\bar\mu - t} - 2\sqrt{\bar\mu}}{t}\right)dt
    \\ & 
   -
  \frac{1}{2\bar\mu}\int_{T b}^{\bar\mu}\left(\frac{1}{\sqrt{\bar\mu + t}} + \frac{1}{\sqrt{\bar\mu - t}}\right)dt + \frac{1}{\sqrt{\bar\mu}}\int_{T b}^{\bar\mu}  \frac{ \Upsilon_0(t)}t dt\,. \label{ccdd}
  \end{align}
The first integral converges to 
$\bar\mu^{-1/2} \ln\left(1+\sqrt
      2\right)$ as ${T, \delta_\mu \to 0}$.
  The second integral gives $\frac{1}{\sqrt{\bar\mu}}\left( -\ln(1+\sqrt{2}) + \sqrt{2} +\ln 4 -2
  \right)$ and the  third integral is $\sqrt{2/\bar\mu}$ in this limit.

To evaluate the last term in (\ref{ccdd}), we proceed as in Lemma~\ref{calcm}, with the obvious modifications. The result is that 
$$ \frac{1}{\sqrt{\bar\mu}}\int_{Tb}^{\bar\mu}
  \frac{ \Upsilon_0(t)}{t}dt = \frac{1}{\sqrt{\bar\mu}}\left( \gamma + \ln{\frac{\bar\mu}{T}} -\ln \frac 8\pi- \kappa^{\rm o}(c)
+ o(1) \right)\,.$$
 Collecting all the terms proves the lemma.
 \end{proof}
This concludes the proof of Theorem~\ref{upperb}. 
\end{proof}

\subsection{Evaluation of $T_{\delta_\mu}^{\rm g}$ and a more detailed phase diagram}

Our final goal is derive a sharper upper bound on the critical
temperature which separates the phases where the minimizer of the BCS
functional has a vanishing or a non-vanishing $\alpha$. We are able to
do this under additional assumptions on the interaction potential $V$,
which, in particular, imply that the BCS minimizer $\alpha$ is the
ground state of the operator $ K^\Delta_{T,\delta_\mu} + \lambda V$.

Recall the definition of $T^{\rm g}_{\delta_\mu}$ in (\ref{tcdef}). 

\begin{theorem}\label{tctheorem}
  Let $V$ be a radial function in $L^{3/2}(\R^3)\cap L^{1}(\R^3)$,
  with $\hat V \leq 0$ and $\hat V(0) < 0$.  Assume that
  $e_\mu=\infspec \V_{\bar\mu}<0$, and let $\varrho(\lambda)$ be
  defined as in (\ref{ddeff}).  Then any sequence of pairs
  $(\delta_\mu(\lambda),T(\lambda))$ on the curve $T^{\rm
    g}_{\delta_\mu}(\lambda V)$ satisfies
\begin{equation}
\lim_{\lambda \to 0}\left( \ln
{\frac{\bar\mu}{T}} + \frac{\pi}{2
\sqrt{{\bar\mu}}\varrho(\lambda)} -
\kappa^{\rm g}(\delta_\mu/T)\right) = 2 - \gamma - \ln (8/\pi) 
\end{equation}
 where $\kappa^{\rm g}(t) = \inf_{d> 0} \zeta(t,d)$, with
 \begin{align}\label{ctc5}
 \zeta(t,d) =&- \int_0^{ d} e^{-x}\ln{(x/d)}\,dx +\int_{ d}^{\infty} e^{-x}\ln{\left(1+\sqrt{1-\left(d/x\right)^2}\right)}dx \\  \nonumber & + \frac{1}{1+e^t}  \int_{d}^{\infty} \frac{1-e^{-x}}{1+e^{x+t}}\,\frac{dx}{\sqrt{x^2 -d^2}} \\ \nonumber &+ \ \frac{1}{1+e^{-t}}  \int_{d}^{\infty}\frac{1-e^{-x}}{1+e^{x-t}} \, \frac{dx}{\sqrt{x^2 -d^2}}  - \ln  \pi \,.
\end{align}
Moreover, for any sequence $({\delta_\mu} (\lambda), T (\lambda))$  that lies  strictly above the curve
$T_{\delta_\mu}^{\rm g}$, in the sense that 
\begin{equation}
\lim_{\lambda \to 0}\left( \ln
{\frac{\bar\mu}{T}} + \frac{\pi}{2
\sqrt{{\bar\mu}}\varrho(\lambda)} -
\kappa^{\rm g}(\delta_\mu/T)\right) < 2 - \gamma - \ln (8/\pi) 
\end{equation}
the BCS gap equation \eqref{overgapeqn} cannot have a solution for small enough $\lambda$  with the property that
the operator $ K^\Delta_{T,\delta_\mu} + \lambda V$
has $0$ as lowest eigenvalue. In particular, the normal state minimizes $\F_T$ in this case.
\end{theorem}

We note that $\lim_{d\to 0} \xi(t,0) = \kappa^{\rm i}(t)$. Moreover,
$\kappa^{\rm g}(t) = \kappa^{\rm i}(t)$ for $t\leq \cosh^{-1}(2)\approx
1.32$. Numerically, one can check that this equality holds even on the
larger interval $t\in [0,1.91]$.  (Compare with Fig.~\ref{fig1}.)

The last statement of the theorem concerning the fact that the normal
state minimizes $\F_T$ follows from a Perron-Frobenius argument, as
explained in Subsect.~\ref{ss:det}. It implies that a non-vanishing
and minimizing $\alpha$ is necessarily the ground state of $
K^\Delta_{T,\delta_\mu} + \lambda V$, with eigenvalue zero.

\begin{proof}
  By the previous arguments we obtain for $(\delta_\mu,T) \in
  T^{\rm g}_{\delta_\mu}$ the following behavior in the small coupling
  limit:
  \begin{equation}\label{tceq1}
  \lim_{\lambda\to 0} \left( \om \left( \delta_\mu, T\right) + \frac 1{\inf {\rm spec} \left
        (\lambda \sqrt{ \bar\mu} \, \V_{\bar\mu} - \lambda^2 \W_{\bar\mu}\right)} \right) =
  0\,,
\end{equation}
where $$\om \left(\delta_\mu,T \right) = \frac 1{4\pi \bar\mu}
\max_{y}\int_{\R^3}\left( \frac 1 {K^{y}_{T,\delta_\mu}(p) } -\frac
  1{p^2}\right) dp\,. $$ For the first part of the theorem it remains
to evaluate $\om\left(\delta_\mu,T \right)$, which is done in the
following lemma.

\begin{Lemma}\label{tclemma}
  In the limit $(T,\delta_\mu) \to (0,0)$,
\begin{equation}\label{tcm1}
\om \left(\delta_\mu,T \right) = \frac 1{\sqrt{\bar\mu}}
\left( \ln \frac {\bar\mu}{T} + \gamma - 2 +\ln \frac 8\pi
-\kappa^{\rm g}(\delta_\mu/T) + o(1)\right)\,.
\end{equation}
\end{Lemma}

\begin{proof}
As in the proof of Lemma~\ref{calcm}, we  can rewrite
 $\om \left(\delta_\mu,T \right)$ as
  \begin{align}\label{tcm7}
 \om \left(\delta_\mu,T \right) =\max_{y}\, I (\delta_\mu, T, y)
 \end{align}
 where
 \begin{align}\label{tcm77}
 I (\delta_\mu, T, y)=& \frac{1}{2\bar\mu}  \int_{\bar\mu}^{\infty} \left( \Upsilon_{y} \left(t \right)  \sqrt{\frac{\bar\mu +t}{t^2 + y^2}}  - \frac{1}{\sqrt{\bar\mu +t}} \right)dt
  \\ \nonumber &  +  \frac{1}{2\bar\mu}  \int^{\bar\mu}_{0} \frac{\Upsilon_{y} \left(t \right) }{\sqrt{t^2 + y^2}} \left(\sqrt{\bar\mu + t} +\sqrt{\bar\mu - t} - 2\sqrt{\bar\mu} \right)dt
  \\ \nonumber & - \frac{1}{2\bar\mu}\int_{0}^{\bar\mu}\left(\frac{1}{\sqrt{\bar\mu + t}} + \frac{1}{\sqrt{\bar\mu - t}}\right)dt  + \frac{1}{\sqrt{\bar\mu}}  \int^{\bar\mu}_{0} \frac{\Upsilon_{y} \left(t\right) }{\sqrt{t^2 + y^2}} dt 
  \end{align}
and 
$$ \Upsilon_{y} \left(t \right)  =
\left(1-\frac{1}{1+e^{\frac{\sqrt{t^2 +y^2} +\delta_\mu}{T}} }-
\frac{1}{1+e^{\frac{\sqrt{t^2 +y^2} -\delta_\mu}{T}} }\right)\,.$$

It is clear that the maximum is attained for a $y$ that goes to zero
as $(T,\delta_\mu)\to (0,0)$.  Using the Lebesgue dominated
convergence theorem, one observes that in the limit $T, y, \delta_\mu
\to 0$ the first integral converges to $\bar\mu^{-1/2}
\ln\left(1+\sqrt 2\right)$, the second becomes
$\bar\mu^{-1/2}\left( -\ln(\sqrt{2} - 1) + \sqrt{2} +\ln 4 -2
\right)$, and the third $\sqrt{2/\bar\mu}$.
  
It remains to compute the last integral in (\ref{tcm77}).  Changing
variables to $x=\sqrt{t^2+y^2}/T$, we can rewrite it as
  \begin{align}\label{tcm8}
  &\int_{d}^{\sqrt{d^2 + (\bar\mu/T)^2}} \frac {1-\frac{1}{1+e^{x+c}} - \frac{1}{1+e^{x-c}}}{\sqrt{x^2 - d^2}}\, dx = I + II := \\ \nonumber
  &= \int_{d}^{\sqrt{d^2+(\bar\mu/T)^2}}  \frac{ \frac{1}{1+e^{c}}- \frac{1}{1+e^{x+c}}}{\sqrt{x^2 -d^2}}\, dx + \int_{d}^{\sqrt{d^2+ (\bar\mu/T)^2}} \frac{ \frac{1}{1+e^{-c}}- \frac{1}{1+e^{x-c}}}{\sqrt{x^2 -d^2}}\, dx\,,
 \end{align}
  where $c=\frac{\delta_\mu}{T}$ and $d=\frac{y}{T}$.
  Now $I$ in (\ref{tcm8}) equals 
  \begin{equation}
  I =  \frac{1}{1+e^{c}} \int_{d}^{\sqrt{d^2 + \left({\bar\mu}/{T}\right)^2}}
  \left[ \frac{1-e^{-x}}{\sqrt{x^2 -d^2}} -
 \frac{1}{\sqrt{x^2 - d^2}}\left(\frac{1-e^{-x}}{1+e^{x+c}}\right)\right]dx
 \label{tcm9} \, ,
  \end{equation}
  and similarly for $II$, replacing $c$ by $-c$. For the second term in the integrand, we can simply replace the upper integration boundary by $\infty$ as $T\to 0$. 
To evaluate the  integral of the first term, we integrate by parts and obtain
  \begin{align}
  & \left(1-e^{-\sqrt{d^2+\left({\bar\mu}/{T}\right)^2 }} \right)  \ln \left(\frac{\bar\mu}{T}+\sqrt{d^2 + \left(  {\bar\mu}/ T \right)^2 } \right) - \left(1-e^{-d}\right)  \ln d\\ \nonumber
  & -  \int_{d}^{\sqrt{d^2+\left({\bar\mu}/{T}\right)^2 }}  e^{-x}\ln \left( x +\sqrt{x^2 - d^2} \right) dx\, .
   \end{align}
The latter integral converges as $T\to 0$. Proceeding in the same way with $II$, we have thus shown that
\begin{align}
 & \int^{\bar\mu}_{0} \frac{\Upsilon_{y} \left(t\right) }{\sqrt{t^2 + y^2}} dt \\ & \nonumber = \ln{\left(\frac{\bar\mu}{T}\right)} +\ln2 - \left(1-e^{-d}\right) \ln d -  \int_{d}^{\infty}  e^{-x}\ln \left( x +\sqrt{x^2 - d^2} \right) dx + o(1) \\ \nonumber  & \quad - \frac{1}{1+e^{c}} \int_{d}^{\infty}
 \frac{1-e^{-x}}{1+e^{x+c}}\, \frac{dx}{\sqrt{x^2 - d^2}} - \frac{1}{1+e^{-c}} \int_{d}^{\infty}
 \frac{1-e^{-x}}{1+e^{x-c}}\, \frac{dx}{\sqrt{x^2 - d^2}}\,.
\end{align}
Combining all the terms, keeping in mind (\ref{euler}),
we arrive at (\ref{tcm1}). 
  \end{proof}

  We now turn to the proof of the second part of
  Theorem~\ref{tctheorem}. The order parameter $\Delta(p)$ is assumed
  to satisfy the BCS gap equation. Under the assumption that $\alpha$
  is the ground state of $K^\Delta_{T,\delta_\mu} + \lambda V$ we can
  use the Birman-Schwinger principle to show, as in \cite[Lemma
  4]{HS}, that
\begin{equation}\label{deltaprop}
\Delta(p) = - g(\lambda) \left(\int_{\Omega_{\bar\mu}} \hat V(p-q )
d\omega(q) + \lambda \nu_\lambda(p)\right),
\end{equation}
with $\|\nu_\lambda\|_{\infty} \leq C$ uniformly in $\lambda$, and
$g(\lambda)$  a normalization constant determined by the
gap equation.
Eq.~(\ref{deltaprop}) can be derived using standard perturbation theory applied
to the operator 
\begin{equation}\label{off}
\lambda m(T,\delta_\mu,\Delta) \F |V|^{1/2} \frac{1}{1+\lambda
V^{1/2} M(T,\delta_\mu,\Delta) |V|^{1/2}} V^{1/2} \F^*\,
\end{equation}
where
\begin{equation}\nonumber
m(T,\delta_\mu, \Delta) =
 \frac 1{4\pi \bar\mu}
\int_{\R^3} \left( \frac {1}{K^{\Delta}_{ T, \delta_\mu}(p)}
  - \frac 1{p^2}\right) dp\, ,
\end{equation}
and $M$ is defined as in \eqref{bigM}, with $K^0_{T,\delta_\mu}$
replaced by $K^\Delta_{T,\delta_\mu}$ and $m(T,\delta_\mu)$ by
$m(T,\delta_\mu,\Delta)$. The eigenvector corresponding to the lowest
eigenvalue $-1$ of this operator is, to leading order, given by the
lowest eigenvector of $\V_{\bar \mu}$, which is the constant function
$u(p)=\frac 1 {\sqrt{4\pi\bar \mu}}$. In fact, the radial symmetry of
$V$ implies that the eigenstates of $\V_{\bar \mu}$ are the spherical
harmonics, and the fact that $\hat V \leq 0$ forces the ground state
of $\V_{\bar \mu}$ to have a fixed sign, hence to be constant. The
lowest eigenvector of (\ref{off}) is, therefore, $\phi = u + \lambda
\xi_\lambda$, with $\xi_\lambda$ uniformly bounded in
$L^2(\Omega_{\bar\mu})$ norm. With the aid of the Birman-Schwinger
principle, $\alpha$ can be recovered from $\phi$ via $\alpha =
c V^{1/2} \F^*\phi$ for some normalization constant
$c$. Hence we have
$$
\Delta(p)= 2
K^\Delta_{T,\delta_\mu}(p) \hat \alpha(p) = - 2 \lambda \widehat {V
  \alpha} (p) = - 2 \lambda c \widehat{|V|^{1/2} \phi} (p)\,,
$$ which implies \eqref{deltaprop}.

Using the Lipschitz continuity of the main term in (\ref{deltaprop}),
we can argue as in \cite[Thm.~2]{HS} that for small $\lambda$ the
function $m \left(\delta_\mu, {T}, \Delta \right)$ is determined by
the value of $\Delta(p)$ at the Fermi surface $p^2 = \bar \mu$, i.e.,
$$
m \left(\delta_\mu, {T}, \Delta
\right) = m \left(\delta_\mu, {T}, \Delta(\sqrt{\bar \mu}) \right) +
o(1)\,.$$ 
This implies further that
\begin{equation}\label{p30}
  \lim_{\lambda\to 0} \left( m \left( \delta_\mu, T, \Delta(\sqrt{\bar\mu}) \right) + \frac 1{\inf {\rm spec} \left
        (\lambda \sqrt{ \bar\mu} \, \V_{\bar\mu} - \lambda^2 \W_{\bar\mu}\right)} \right) =
  0\,.
\end{equation}
Since, by definition, 
$$  
m \left(\delta_\mu, {T}, \Delta(\sqrt{\bar\mu}) 
\right) \leq \overline m(\delta_\mu, {T})\,, $$
it follows from (\ref{tceq1}) that (\ref{p30}) cannot be satisfied for $(\delta_\mu,T)$ outside of
$T^{\rm g}_{\delta_\mu}$, i.e., for such $(\delta_\mu,T)$ such that 
\begin{equation}
  \lim_{\lambda\to 0} \left( \overline m \left( \delta_\mu, T \right) + \frac 1{\inf {\rm spec} \left
        (\lambda \sqrt{ \bar\mu} \, \V_{\bar\mu} - \lambda^2 \W_{\bar\mu}\right)} \right) <  0\,.
\end{equation}
This completes the proof of the theorem.
\end{proof}


\end{document}